\documentclass[letterpaper, 10 pt, conference]{ieeeconf}  

\IEEEoverridecommandlockouts                              

\usepackage{amsmath,amsfonts,amssymb}
 
\usepackage{amsthm}

\newtheorem{theorem}{Theorem}
\newtheorem{definition}{Definition}
\newtheorem{proposition}{Proposition}
\newtheorem{example}{Example}
\newtheorem{remark}{Remark}
\usepackage{cite}
\usepackage{algorithmic}
\usepackage{textcomp}
\usepackage{graphicx}
\usepackage{courier}
\usepackage{mathtools}
\usepackage{float}
\DeclarePairedDelimiter{\norm}{\lVert}{\rVert} 

\def\BibTeX{{\rm B\kern-.05em{\sc i\kern-.025em b}\kern-.08em
		T\kern-.1667em\lower.7ex\hbox{E}\kern-.125emX}}
\usepackage[font=small,skip=8pt]{caption}
\DeclareMathAlphabet{\mathpzc}{OT1}{pzc}{m}{it}
\usepackage{mathrsfs}
\makeatletter
\let\NAT@parse\undefined
\makeatother
\usepackage{hyperref}  
\usepackage{tikz}

\title{\LARGE \bf
	Data-Based System Analysis and Control of Flat Nonlinear Systems
}

\author{Mohammad Alsalti$^{1}$, Julian Berberich$^{2}$, Victor G. Lopez$^{1}$, Frank Allgöwer$^{2}$, and Matthias A. Müller$^{1}$ %
	\thanks{$^{1}$Leibniz University Hannover, Institute of Automatic Control, 30167 Hannover, Germany. E-mail:\{\href{maitlo:alsalti@irt.uni-hannover.de}{alsalti}, \href{maitlo:mueller@irt.uni-hannover.de}{mueller}, \href{maitlo:lopez@irt.uni-hannover.de}{lopez}\}@irt.uni-hannover.de}%
	\thanks{$^{2}$University of Stuttgart, Institute for Systems Theory and Automatic Control, 70550 Stuttgart, Germany. E-mail: \{\href{maitlo:julian.berberich@ist.uni-stuttgart.de}{julian.berberich}, \href{maitlo:frank.allgower@ist.uni-stuttgart.de}{frank.allgower}\}@ist.uni-stuttgart.de}%
	\thanks{This work has received funding from the European Research Council (ERC) under the European Union’s
		Horizon 2020 research and innovation programme (grant agreement No 948679). This work was also funded by Deutsche Forschungsgemeinschaft (DFG, German Research Foundation) under Germany’s Excellence Strategy - EXC 2075 - 390740016. The authors thank the International Max Planck Research School for Intelligent Systems (IMPRS-IS) for supporting Julian Berberich.}
}

\newcommand\copyrighttext{%
	\footnotesize \copyright 2021 IEEE. Personal use of this material is permitted. Permission from IEEE must be obtained for all other uses, in any current or future media, including reprinting/republishing this material for advertising or promotional purposes, creating new collective works, for resale or redistribution to servers or lists, or reuse of any copyrighted component of this work in other works.}
\newcommand\copyrightnotice{%
	\begin{tikzpicture}[remember picture,overlay]
		\node[anchor=south,yshift=10pt] at (current page.south) {\fbox{\parbox{\dimexpr\textwidth-\fboxsep-\fboxrule\relax}{\copyrighttext}}};
	\end{tikzpicture}%
}

\begin{document}
	\maketitle
	\thispagestyle{empty}
	\pagestyle{empty}
	\copyrightnotice	
	\begin{abstract}
		Willems et al. showed that all input-output trajectories of a discrete-time linear time-invariant system can be obtained using linear combinations of time shifts of a single, persistently exciting, input-output trajectory of that system. In this paper, we extend this result to the class of discrete-time single-input single-output flat nonlinear systems. We propose a data-based parameterization of all trajectories using only input-output data. Further, we use this parameterization to solve the data-based simulation and output-matching control problems for the unknown system without explicitly identifying a model. Finally, we illustrate the main results with numerical examples.
	\end{abstract}
	\section{Introduction}
	Modeling complex systems from first principles can in general be a difficult task. This is why using data to identify a mathematical model of the system has been studied for decades \cite{Ljung87}. Typically, one would identify a model from data, and later use it to design controllers and analyze its stability and robustness properties. Accurate models are often difficult to obtain and may lead to complex controller structures. Conversely, inaccurate models (due to unmodeled dynamics or the use of noisy data in the identification process) may degrade robustness and stability of the control system \cite{Zhong13}. In recent years, there has been an increasing interest in analyzing systems and designing controllers from process data directly, without explicitly identifying a model of the system. Examples of purely data-driven analysis and control techniques are iterative learning control \cite{Bristow06}, virtual reference feedback tuning \cite{Campi02}, unfalsified control theory \cite{Safonov}, among many others \cite{Zhong13,Tanaskovic17, Bazanella12}.\par It was shown in \cite{Willems05} that for the class of discrete time (DT) linear time-invariant (LTI) systems, \textit{all} input-output trajectories of a system can be parameterized using linear combinations of time-shifts of a single, \textit{persistently exciting}, input-output trajectory of that system. This result is now commonly referred to as the \textit{fundamental lemma}, and was applied to solve problems such as data-based simulation and control \cite{Markovsky08}. While originally being developed in the behavioral setting \cite{Willems05}, this result has recently been reformulated in the standard state-space framework \cite{Berberich20, vanWaarde20}. In \cite{Persis20,Berberich202}, the fundamental lemma was used to design (robust) state feedback and LQR controllers from data. Moreover, purely data-driven MPC schemes have been developed in \cite{Coulson20} and were analyzed for stability and robustness guarantees in \cite{Berberich203}.\par For nonlinear systems, the fundamental lemma was extended to the class of Hammerstein-Wiener systems in \cite{Berberich20} and second order discrete Volterra systems in \cite{Escobedo20}. Stabilization of single-input single-output (SISO) feedback linearizable systems using approximate models appeared in \cite{Tabuada20}. There, input-affine continuous-time systems were addressed assuming constant intersampling behavior of the states under high enough sampling rate. Other data-driven stabilization techniques for some classes of nonlinear systems appeared in \cite{Guo20, Strasser20}. In this paper, we consider the class of flat nonlinear systems \cite{Fliess92}, also referred to as difference flat nonlinear systems in discrete time \cite{Guillot19}. Flat systems are an important class of nonlinear systems that describe many practical applications such as robotic arms, crane control, and automatic flight control \cite{Levine09}.\par
	The contributions of this paper are the following. In Section III, we extend Willems' fundamental lemma to difference flat nonlinear systems,
	by leveraging the fact that a SISO flat system is fully feedback linearizable and by using a set of basis functions that describe the system's nonlinearities. In Section IV, we use this result to solve the data-based simulation and the output-matching control problems for flat nonlinear systems in a data-based fashion, without explicitly identifying a model. Finally, we illustrate the main results using numerical examples in Section V.
	\section{Preliminaries}
	For a sequence $\{z_k\}_{k=0}^{N-1}$ with $z_k\in\mathbb{R}^\sigma$, we use the following notation to denote the stacked sequence $z = \begin{bmatrix}z_0^\top &z_1^\top & \dots & z_{N-1}^\top\end{bmatrix}^\top$. Further, $z_{[l,j]}$ will be used to denote a stacked window of that sequence, i.e., $z_{[l,j]} = \begin{bmatrix}z_l^\top &z_{l+1}^\top & \dots & z_{j}^\top\end{bmatrix}^\top$ with $0\leq l<j\in\mathbb{N}$. We define the Hankel matrix and its block rows as
	\begin{equation*}
		\hspace{-1mm}H_L(z)\hspace{-0.5mm}=\hspace{-1mm}\begin{bmatrix}\hspace{-0.5mm}z_0 & z_1 & \dots & z_{N-L}\\
			\hspace{-0.5mm}z_1 & z_2 & \dots & z_{N-L+1}\hspace{-0.5mm}\\
			\vdots & \vdots & \ddots & \vdots\\
			\hspace{-0.5mm}z_{L-1} & z_L & \dots & z_{N-1}\end{bmatrix}\hspace{-1mm}=\hspace{-1mm}\begin{bmatrix}{h}_{0}(z_{[0,N-L]})\\{h}_{1}(z_{[1,N-L+1]})\\\vdots\\\hspace{-0.25mm}{h}_{L-1}(z_{[L-1,N-1]})\hspace{-0.25mm}\end{bmatrix}\hspace{-0.5mm}.
	\end{equation*}
	For the discussion in this paper, persistency of excitation (PE) is defined as follows.
	\begin{definition} The sequence \(\{z_k\}_{k=0}^{N-1}\) is said to be persistently exciting of order \(L\) if \(\textup{rank}(H_{L}(z))=\sigma L\).
	\end{definition}
	In the following, we recall the main result of \cite{Willems05} in the state-space framework (compare \cite{Berberich20, vanWaarde20}). Consider the following DT-LTI system of the form
	\begin{equation}
		\begin{aligned}
			x_{k+1}=Ax_k+Bu_k,\\
			y_k=Cx_k+Du_k,
		\end{aligned}
		\label{DTLTI}
	\end{equation}
	where \(x_k\in\mathbb{R}^n\) is the state, \(u_k\in\mathbb{R}^m\) is the input, \(y_k\in\mathbb{R}^p\) is the output and the pair \((A,B)\) is controllable. The fundamental lemma is stated as follows.\par
	\begin{theorem}[\hspace{-0.25mm}\cite{Willems05}, Theorem 1]
		Let \(\{u_k,y_k\}_{k=0}^{N-1}\) be an input-output trajectory of \eqref{DTLTI}. If \(\{u_k\}_{k=0}^{N-1}\) is persistently exciting of order \((L+n)\), then any \(\{\bar{u}_k,\bar{y}_k\}_{k=0}^{L-1}\) is a trajectory of (1), if and only if there exists \(\alpha\in\mathbb{R}^{N-L+1}\) such that
		\begin{equation}
			\begin{bmatrix} H_L(u)\\ H_L(y)\end{bmatrix} \alpha = \begin{bmatrix}\bar{u} \\ \bar{y}\end{bmatrix}.
			\label{fundamental_lemma}
		\end{equation}
	\end{theorem}\par
	In other words, Theorem 1 states that the vector space of all input-output trajectories of a DT-LTI system can be spanned by time shifts of previously collected, persistently exciting, input-output data. Persistency of excitation of order \((L+n)\) implies that knowledge of an upper bound on the system order is needed, and that the length of the a priori collected data $\{u_k,y_k\}_{k=0}^{N-1}$ is $N\geq(m+1)(L+n)-1$.
	\section{Data-Based Representation of Flat Nonlinear Systems}
	Differential flatness was described in continuous time in the seminal work by Fliess et al. \cite{Fliess92}. In DT, this concept is referred to as difference flatness and corresponds to a class of nonlinear systems that have an output (called flat output) of the same dimension as the input, such that all of the system's variables can be expressed as a function of that output and a finite number of its successive forward-time shifts \cite{Guillot19}. \par
	In this section, we extend the results of Theorem 1 to the class of DT-SISO flat nonlinear systems. In Section III.A, we give an overview of the DT relative degree and DT feedback linearization. Later, in Section III.B, we make use of the property that SISO flat nonlinear systems are fully feedback linearizable \cite[Theorem 1]{Diwold21} and propose a data-based parametrization of their trajectories.
	\subsection{Discrete-Time Feedback Linearization}
	Feedback linearization is a common nonlinear control technique which transforms the nonlinear system to a fully or partially linearized system \cite{Isidori95}. In particular, a feedback control law and a coordinate transformation linearize the input-output map and split the dynamics into a linear (external) and nonlinear (internal) pair. If the internal dynamics are stable, then a control input can be synthesized to cancel the nonlinearities, thereby allowing for application of linear control design techniques.\par
	We now recall the notion of the DT relative degree and DT feedback linearization. Consider the following DT-SISO nonlinear system
	\begin{equation}
		\begin{aligned}
			x_{k+1} &= f(x_k,u_k),\\
			y_k &= h(x_k),
		\end{aligned}\label{NLsys}
	\end{equation}
	with \(x_k\in\mathbb{R}^n,\, u_k\in\mathbb{R},\,y_k\in\mathbb{R}\) and \(f: \mathbb{R}^n\times\mathbb{R}\to\mathbb{R}^n,\,h:\mathbb{R}^n\to\mathbb{R}\) being smooth functions with $f(0,0)=0$ and $h(0)=0$. We define \(f_O^j\) to be the \(j^{\text{th}}\) iterated composition of the undriven dynamics \(f(\cdot,0)\). The system \eqref{NLsys} is said to have a (globally) well-defined relative degree \(d\) 
	\cite{Monaco87, Monaco92}, if
	\begin{equation}
		\begin{aligned}
			\frac{\partial\left(h\left(f_O^k\left(f(x,u)\right)\right)\right)}{\partial u}&=0, \quad \text{ for }0\leq k<d-1,\\
			\frac{\partial\left(h\left(f_O^{d-1}\left(f(x,u)\right)\right)\right)}{\partial u}&\neq 0, \quad \forall x\in \mathbb{R}^n,\, \forall u\in \mathbb{R},
			\label{rel_deg}
		\end{aligned}
	\end{equation}
	holds. In addition, we assume throughout the paper that
	\begin{equation}
		0\in\text{ Im}\left(h(f_O^{d-1}(f(x,\cdot)))\right),\quad \forall x\in \mathbb{R}^n.
		\label{img}
	\end{equation}
	Condition \eqref{img} means that $0$ lies in the image of $h(f_O^{d-1}(f(x,\cdot)))$ and is needed to transform the system to the normal form in \eqref{BINF} below (compare \cite{Monaco87}).\par 
	\begin{remark}
		If the relative degree is only well-defined locally around an equilibrium point ($x_e,u_e$), then the results of Theorem 2 hold locally in a neighborhood of that point. The main results of this paper (Prop. 1-3) remain true for all input-output trajectories of system \eqref{NLsys} in the corresponding neighborhoods where Theorem 2 holds.
	\end{remark}
	\indent It follows from \eqref{rel_deg} that
	\begin{equation}
		\begin{aligned}
			y_{k+i} &= h(f_O^i(x_k)), \quad 0\leq i\leq d-1,\\
			y_{k+d} &= h(f_O^{d-1}(f(x_k,u_k))).
			\label{y_def}
		\end{aligned}
	\end{equation}
	This means that either the relative degree is finite ($d\leq n$) and \(y_{d}\) is the first output to be affected by the initial input \(u_0\), or it is infinite and the output is never affected by the input \cite{Monaco87}. Using the above definition of the DT relative degree, feedback linearization of DT systems is given as follows.\par
	\begin{theorem}[\hspace{-0.25mm}\cite{Monaco87}, Proposition 2.1]
		Let \eqref{NLsys} have a globally well-defined relative degree \(d\leq n\). Then there exists a static nonsingular state feedback control law \(u=\gamma(x,v)\) and a change of coordinates \(z=\begin{bmatrix}\xi^\top&\eta^\top\end{bmatrix}^\top=\begin{bmatrix}T_1^\top(x) & T_2^\top(x)\end{bmatrix}^\top\) such that the system \eqref{NLsys} can be brought to the following form
		\begin{equation}
			\begin{aligned}
				\xi_{k+1} &= A_c \xi_k + B_cv_k,\\
				\eta_{k+1} &= F\left(z_k,v_k\right),\\
				y_k &= C_c \xi_k,
			\end{aligned}
			\label{BINF}
		\end{equation}
		with $\xi_k\in\mathbb{R}^d,\,\eta_k\in\mathbb{R}^{n-d},\,v_k\in\mathbb{R},\,\gamma : \mathbb{R}^n \times \mathbb{R} \to \mathbb{R},\,{\partial \gamma}/{\partial v}\neq0,\forall (x,v)\in\mathbb{R}^n\times\mathbb{R},\,T_1:\mathbb{R}^n\to\mathbb{R}^{d}$, $T_2:\mathbb{R}^n\to\mathbb{R}^{n-d}$, $F:\mathbb{R}^n\times\mathbb{R}\to\mathbb{R}^{n-d}$ and \(A_c,B_c,C_c\) in Brunovsky canonical form
		\footnote{The Brunovsky canonical form has the following structure \begin{equation*}A_c=\begin{bmatrix}0&1&\dots&0 \\ \vdots&\ddots&\ddots&\vdots \\ \vdots& &\ddots&1 \\ 0&\dots&\dots&0\end{bmatrix}, \quad B_c=\begin{bmatrix}0\\ \vdots\\ 0\\ 1\end{bmatrix},\quad C_c^\top=\begin{bmatrix}1\\ 0\\ \vdots\\ 0\end{bmatrix}\end{equation*}}
		\cite{Zeitz89}, which are a controllable/observable triplet.\end{theorem}
	\indent In \eqref{BINF}, \(v_k\) is an external (synthetic) control input and $\eta_k$ are the internal dynamics, whereas $\xi_k$ are the external dynamics given by time shifts of the output
	\begin{equation}
		\begin{aligned}
			\xi_k&\vcentcolon=\begin{bmatrix}h(x_k)&h\left(f_O(x_k)\right)&\dots&h\left(f_O^{d-1}(x_k)\right)\end{bmatrix}^\top,\\
			&\stackrel{\eqref{y_def}}{=} \begin{bmatrix}y_k&y_{k+1}&\dots&y_{k+d-1}\end{bmatrix}^\top=y_{[k,k+d-1]}.
		\end{aligned}
		\label{xi}
	\end{equation} 
	\indent In general, the internal dynamics are not necessarily linear, posing additional challenges to the controller design in case they are unstable. However, a main advantage of SISO flat systems is that they are fully feedback linearizable \cite[Theorem 1]{Diwold21}, which implies that their relative degree is equal to the system dimension (i.e., $d=n$) and that there are no internal dynamics. Therefore, the system in \eqref{NLsys} can be be brought to the following form
	\begin{equation}
		\begin{aligned}
			\xi_{k+1} &= A_c \xi_k + B_cv_k,\\
			y_k &= C_c \xi_k.
		\end{aligned}\label{flat_NL}
	\end{equation}
	\indent We show in the next section how the system representation in \eqref{flat_NL} is used to extend Theorem 1 to SISO flat systems.
	\subsection{Extension of the Fundamental Lemma to Flat Systems}
	According to Theorem 1, the entire vector space of input-output trajectories of a DT-LTI controllable system can be spanned using time-shifts of a single, persistently exciting, input-output trajectory. The intuition here is that for \textit{any} system that is linear in suitably chosen coordinates and is controllable, the ideas of Theorem 1 apply. 
	By invertibility of $\gamma(x,v)$ w.r.t. $v$ as guaranteed by Theorem 2, one can write \(v_k=\Phi(u_k,T_1^{-1}(\xi_k))\) for some $\Phi:\mathbb{R}\times\mathbb{R}^n\to\mathbb{R}$. Furthermore, we assume that \(\Phi\) admits a basis function expansion of the form
	\begin{equation}
		\begin{aligned}
			v_k&=\Phi(u_k,T_1^{-1}(\xi_k)) = \sum\limits_{i=1}^{r} a_i{\psi}_i(u_k,\xi_k)\\
			&\eqqcolon \text{a}^\top \Psi(u_k,\xi_k)\stackrel{\eqref{xi}}{=}
			\text{a}^\top\Psi(u_k,y_{[k,k+n-1]})
			,\label{basis}
		\end{aligned}
	\end{equation}
	where \(\Psi(u_k,y_{[k,k+n-1]})\) is the stacked vector of $r\in\mathbb{N}$ linearly independent basis functions ${\psi}_i: \mathbb{R}\times\mathbb{R}^n\to\mathbb{R}$, and $\textup{a}^\top$ is the vector of coefficients \(a_i\) which are not all zero, for $i\in\mathbb{Z}_{[1,r]}$. For notational convenience, we define $\hat{\Psi}_k(u,y)\coloneqq \Psi(u_k,y_{[k,k+n-1]})$.
	Finally, by substituting \eqref{basis} into \eqref{flat_NL} we obtain
	\begin{equation}
		\begin{aligned}
			\xi_{k+1} &= A_c\xi_k+{B}_c\text{a}^\top\hat{\Psi}_k(u,y),\\
			y_k &= C_c\xi_k.
		\end{aligned}
		\label{flat_linear}
	\end{equation}
	The following proposition extends the result of Theorem 1 to SISO flat nonlinear systems.
	\begin{proposition} Suppose \(\{u_k\}_{k=0}^{N-n-1},\{y_k\}_{k=0}^{N-1}\) is a trajectory of a flat system as in \eqref{NLsys}, and that $\{\hat{\Psi}_k(u,y)\}_{k=0}^{N-n-1}$ from \eqref{flat_linear} is persistently exciting of order \(L\).
		Then, \(\{\bar{u}_k\}_{k=0}^{L-n-1},\{\bar{y}_k\}_{k=0}^{L-1}\) is a trajectory of the system \eqref{NLsys} if and only if there exists \(\alpha\in\mathbb{R}^{N-L+1}\) such that
		\begin{gather}
			\begin{bmatrix} H_{L-n}(\hat{\Psi}(u,y))\\ H_{L}(y)\end{bmatrix} \alpha = \begin{bmatrix}\hat{\Psi}(\bar{u},\bar{y})\\ \bar{y}\end{bmatrix},
			\label{flat_DB}
		\end{gather}
		with $\hat{\Psi}(u,y),\hat{\Psi}(\bar{u},\bar{y})$ being the stacked vectors of $\{\hat{\Psi}_k(u,y)\}_{k=0}^{N-n-1},\,\{\hat{\Psi}_k(\bar{u},\bar{y})\}_{k=0}^{L-n-1}$ respectively.
	\end{proposition}
	\begin{proof}
		As seen from \eqref{y_def}, the effect of an input sequence $\{u_k\}_{k=0}^{N-n-1}$ can be observed from the corresponding output sequence $\{y_k\}_{k=0}^{N-1}$.
		According to Theorem 2, if \(\{\hat{\Psi}_k(\bar{u},\bar{y})\}_{k=0}^{L-n-1},\{\bar{y}_k\}_{k=0}^{L-1}\) is a trajectory of \eqref{flat_linear}, then \(\{\bar{u}_k\}_{k=0}^{L-n-1},\{\bar{y}_k\}_{k=0}^{L-1}\) is a trajectory of \eqref{NLsys}.\\
		\indent The pair $(A_c,B_c)$ is controllable by definition, i.e., $\begin{bmatrix}B_c & A_cB_c  & \dots & A_c^{n-1}B_c\end{bmatrix}$ has full rank. Then, $\begin{bmatrix}{B}_c\text{a}^\top & A_c{B}_c\text{a}^\top & \dots & A^{n-1}_c{B}_c\text{a}^\top\end{bmatrix}$ also has full rank since not all $a_i$ in \eqref{basis} are zero, implying that the pair $(A_c,{B}_c\text{a}^\top)$ is also controllable. Since $\{\hat{\Psi}_k(u,y)\}_{k=0}^{N-n-1}$ 
		is persistently exciting of order $L$ and $(A_c,{B}_c\text{a}^\top)$ is controllable, Theorem 1 implies that \(\{\hat{\Psi}_k(\bar{u},\bar{y})\}_{k=0}^{L-n-1},\{\bar{y}_k\}_{k=0}^{L-n-1}\) is a trajectory of \eqref{flat_linear} if and only if there exists \(\alpha\in\mathbb{R}^{N-L+1}\) such that
		\begin{gather}
			\begin{bmatrix} H_{L-n}(\hat{\Psi}(u,y))\\ H_{L-n}(y_{[0,N-n-1]})\end{bmatrix} \alpha = \begin{bmatrix}\hat{\Psi}(\bar{u},\bar{y})\\ \bar{y}_{[0,L-n-1]}\end{bmatrix}.
			\label{willems_flat}
		\end{gather}
		\indent Since \eqref{flat_NL} is in Brunovsky canonical form, the output can be written as 
		$y_{k+n}=v_{k}\stackrel{\eqref{basis}}{=}\text{a}^\top\hat{\Psi}_k(u,y)$. It follows that $h_{k+n}(y_{[k+n,k+n+N-L]}) = \text{a}^\top h_{k}(\hat{\Psi}_{[k,k+N-L]}(u,y))$ for  $k\in\mathbb{Z}_{[0,L-n-1]}$. From \eqref{willems_flat}, it holds that $h_i(\hat{\Psi}_{[i,i+N-L]}(u,y))\alpha = \hat{\Psi}_i(\bar{u},\bar{y}),\,\forall i\in\mathbb{Z}_{[L-2n,L-n-1]}$. Hence, $\bar{y}_{i+n} = \text{a}^\top\hat{\Psi}_i(\bar{u},\bar{y}) = h_{i+n}(y_{[i+n,i+n+N-L]})\alpha$. Concatenating these $n-$rows with \eqref{willems_flat} yields \eqref{flat_DB}.
	\end{proof}\par
	In order to directly apply Proposition 1, it is required that $v_k=\Phi(u_k,T_1^{-1}(\xi_k))$ can be parameterized by a finite number of known basis functions, for which the coefficients $a_i$, are unknown. This might not be true in practice, so one can choose sufficiently many basis functions to approximate the nonlinearity. However, Proposition 1 requires $\hat{\Psi}$ to be persistently exciting of order $L$ (i.e., rank\((H_{L}(\hat{\Psi}(u,y))) = rL\)). Therefore, choosing many basis functions makes the satisfaction of the PE condition more difficult. Specifically, the length of available data has to be $N\geq(r+1)L+n-1$. Moreover, the PE condition can in general only be checked after collecting both input-output data which is needed to construct $H_{L-n}(\hat{\Psi}(u,y))$, unlike the LTI case where the PE condition can be satisfied by a suitable choice of the input $u$. Since satisfying the PE condition can in general be a difficult task, 
	we show in Section IV how a kernel method can be exploited to implicitly use an infinite number of basis functions, on the expense of only approximating a subset of (rather than characterizing all) input-output trajectories.\par
	Theorem 1 has been widely used in the development of data-driven control methods for linear systems. Proposition 1 provides an extension of the fundamental lemma to the class of SISO flat nonlinear systems, by exploiting linearity properties in suitably chosen coordinates. In the following sections, we will show how this result can be employed for data-based simulation and control.
	\section{Data-Based Simulation and Control}
	\indent In this section, we show how the data-based representation in Proposition 1 allows us to both simulate and control trajectories of unknown flat nonlinear systems using only measured data. Compared to \cite{Markovsky08, Berberich20}, where these problems were considered for linear and (partially) for Hammerstein systems, the main challenge is the presence of the basis functions in \eqref{flat_DB} which depend on both the input $u_k$ and the state $\xi_k$ of the system \eqref{flat_linear}. This necessitates certain modifications and extensions as explained in the following propositions. We first start by addressing the data-based simulation problem where the goal is to compute an unknown system's response to a given input and initial conditions directly from data.
	\begin{proposition}
		Suppose \(\{u_k\}_{k=0}^{N-n-1},\{y_k\}_{k=0}^{N-1}\) is a trajectory of a flat system as in \eqref{NLsys}. Furthermore, assume that
		$\{\hat{\Psi}_k(u,y)\}_{k=0}^{N-n-1}$ from \eqref{flat_linear} is persistently exciting of order \(L\) and let \(\{\bar{u}_k\}_{k=0}^{L-n-1}\) be a new input to \eqref{NLsys} with \(\{\bar{y}_k\}_{k=0}^{n-1}\) specifying initial conditions for the state $\bar{\xi}_0$ in \eqref{flat_linear}. Then there exists \(\alpha\in\mathbb{R}^{N-L+1}\) such that
		\begin{equation}
			\begin{bmatrix}H_{L-n}(\hat{\Psi}(u,y)) \\ H_n(y_{[0,N-L+n-1]}) \end{bmatrix}\alpha=\begin{bmatrix}\hat{\Psi}(\bar{u},H_L(y)\alpha)\\\bar{y}_{[0,n-1]}\end{bmatrix}
			\label{sim}
		\end{equation}
		holds, where $\hat{\Psi}(\bar{u},H_L(y)\alpha)$ is the stacked vector of the sequence $\{\hat{\Psi}_k(\bar{u},H_L(y)\alpha)\}_{k=0}^{L-n-1}$ with $\hat{\Psi}_k(\bar{u},H_L(y)\alpha)=\Psi(\bar{u}_k,H_n(y_{[k,k+N-L+n-1]})\alpha)$, and $\hat{\Psi}(u,y)$ as in Prop. 1. It further holds that $\bar{y}=H_{L}(y)\alpha$.
	\end{proposition}
	\begin{proof}
		Since $\{\bar{y}_k\}_{k=0}^{n-1}$ specifies the initial condition $\bar{\xi}_0$, it follows from Proposition 1 that there exists $\alpha\in\mathbb{R}^{N-L+1}$ such that
		\begin{equation}
			\begin{bmatrix}H_{L-n}(\hat{\Psi}(u,y)) \\ H_n(y_{[0,N-L+n-1]}) \end{bmatrix}\alpha=\begin{bmatrix}\hat{\Psi}(\bar{u},\bar{y})\\\bar{y}_{[0,n-1]}\end{bmatrix}
			\label{unsolvable}
		\end{equation}
		holds, where $\hat{\Psi}(\bar{u},\bar{y})$ is the stacked vector of the sequence $\{\hat{\Psi}_k(\bar{u},\bar{y})\}_{k=0}^{L-n-1}$.
		This is true since \eqref{unsolvable} is a subset of the equations of \eqref{flat_DB}. It further follows from Proposition 1 that $\bar{y}= H_{L}(y)\alpha$, where $\bar{y}$ is the stacked vector of the corresponding output sequence $\{\bar{y}_k\}_{k=0}^{L-1}$.
		Since $\bar{y}$ in \eqref{unsolvable} is unknown, we substitute $\bar{y}=H_{L}(y)\alpha$ and obtain the equivalent equation \eqref{sim}.
	\end{proof}\par
	\indent Since in practice, it is inevitable that collected data is noisy, we use a regularized least-squares optimization problem to solve \eqref{sim} for $\alpha$ (compare \cite{Berberich20})
	\begin{equation}
		\min\limits_{\alpha}\norm*{\hspace{-0.5mm}\begin{bmatrix}H_{L-n}(\hat{\Psi}(u,y)) \\ H_n(y_{[0,N-L+n-1]})\hspace{-0.75mm} \end{bmatrix}\hspace{-1mm}\alpha\hspace{-0.5mm}-\hspace{-1.5mm}\begin{bmatrix}\hat{\Psi}(\bar{u},H_L(y)\alpha)\\\bar{y}_{[0,n-1]}\end{bmatrix}\hspace{-0.5mm}}_2^2\hspace{-2mm}+\hspace{-0.5mm}\lambda\norm*{\alpha}_2^2.
		\label{sim_opt}
	\end{equation}
	\begin{remark}
		If the basis functions ${\psi}_i$ are affine in $\xi$, problem \eqref{sim_opt} reduces to a quadratic program. Otherwise, it is a nonlinear optimization problem.
	\end{remark}
	Notice that \eqref{sim_opt} does not explicitly depend on the choice of the basis functions, but rather on their scalar product. Therefore, one can (implicitly) use an infinite number of basis functions (i.e., $r=\infty$), and employ a kernel method \cite{Hofmann08} to approximate the scalar multiplication of the infinite stacked vectors of basis functions 
	found in \eqref{sim_opt} (see, e.g., \cite{Berberich20}). 
	For instance, we can use bivariate Gaussian kernel functions of the following form, for some $\sigma>0$ \begin{subequations}
		\begin{align}
			&\Psi(u_i,y_{[i,i+n-1]})^\top \Psi(u_j,y_{[j,j+n-1]})=\notag\\ &\quad\qquad e^{-\frac{\norm*{\small\begin{bmatrix}u_i\\y_{[i,i+n-1]}\end{bmatrix}-\begin{bmatrix}u_j\\y_{[j,j+n-1]}\end{bmatrix}}_2^2}{2\sigma^2}},\\
			&\Psi(\bar{u}_i,H_n(y_{[i,i+N-L+n-1]})\alpha)^\top\Psi(u_j,y_{[j,j+n-1]})=\label{kernelB}\notag\\ &\quad\qquad e^{-\frac{\norm*{\small\begin{bmatrix}\bar{u}_i\\H_n(y_{[i,i+N-L+n-1]})\alpha\end{bmatrix}-\begin{bmatrix}u_j\\y_{[j,j+n-1]}\end{bmatrix}}_2^2}{2\sigma^2}}.
		\end{align}
		\label{kernels}
	\end{subequations}
	\indent As in Proposition 2, we have replaced the unknown $\bar{y}_i$ by $h_{i}(y_{[i,i+N-L]})\alpha$ in \eqref{kernelB}, which results in \eqref{sim_opt} being a nonlinear optimization problem. It should be noted that using an infinite number of unknown basis functions would require an infinitely long persistently exciting input sequence (since rank\((H_{L}(\hat{\Psi}(u,y))) = rL\) is required), which is not attainable in practice. Hence, if kernels are used to compute the inner product in \eqref{sim_opt}, then typically only a subset of the input-output trajectory space can be approximated. In Section V, we illustrate this result with a numerical example.\par
	We now consider the data-based output matching problem for flat nonlinear systems. The goal here is to compute a control input that is required to make an unknown system track a given reference trajectory with given initial conditions.
	\begin{proposition}
		Suppose \(\{u_k\}_{k=0}^{N-n-1},\{y_k\}_{k=0}^{N-1}\) is a trajectory of a flat system as in \eqref{NLsys} and that
		$\{\hat{\Psi}_k(u,y)\}_{k=0}^{N-n-1}$
		from \eqref{flat_linear} is persistently exciting of order \(L\). Let \(\{\bar{y}_k\}_{k=0}^{L-1}\) be a desired reference trajectory with \(\{\bar{y}_k\}_{k=0}^{n-1}\) specifying initial conditions for the state $\bar{\xi}_0$ in \eqref{flat_linear}. Moreover, let one of the basis functions ${\psi}_i$ be the identity\footnote{It is sufficient that one of the basis functions is invertible w.r.t. u. For notational simplicity, we do not explicitly consider this case here.}, i.e., ${\psi}_i(u,\xi)=u$. Then there exists \(\alpha\in\mathbb{R}^{N-L+1}\) such that
		\begin{equation}
			\begin{bmatrix}H_{L-n}(\hat{\Psi}(u,y)) \\ H_L(y) \end{bmatrix}\alpha=\begin{bmatrix}\hat{\Psi}(H_{L-n}(u)\alpha,\bar{y})\\\bar{y}\end{bmatrix}
			\label{flat_OM_eq}
		\end{equation}
		holds, where $\hat{\Psi}(H_{L-n}(u)\alpha,\bar{y})$ is the stacked vector of the sequence $\{\hat{\Psi}_k(H_{L-n}(u)\alpha,\bar{y})\}_{k=0}^{L-n-1}$ with $\hat{\Psi}_k(H_{L-n}(u)\alpha,\bar{y})=\Psi(h_k(u_{[k,k+N-L]})\alpha,\bar{y}_{[k,k+n-1]})$ and $\hat{\Psi}(u,y)$ as in Prop. 1. It further holds that $\bar{u}=H_{L-n}(u)\alpha$.
	\end{proposition}
	\begin{proof}
		Since $\{\bar{y}_k\}_{k=0}^{n-1}$ specifies the initial condition $\bar{\xi}_0$, it can be seen that any reference trajectory $\{\bar{y}\}_{k=0}^{L-1}$ is a trajectory of \eqref{NLsys} given $\{\bar{y}\}_{k=0}^{n-1}$ by choosing a suitable input sequence. In particular, since $\bar{y}_{k+n}={v}_{k}$ holds for $k\in\mathbb{Z}_{[0,L-n-1]}$ as in \eqref{flat_NL}, then according to Theorem 2, the required input to achieve the desired output $\bar{y}$ is given by $\bar{u}_k=\gamma(T_1^{-1}(\bar{\xi}_k),\bar{y}_{k+n})$ for $k\in\mathbb{Z}_{[0,L-n-1]}$.
		Hence, it follows from Proposition 1 that there exists $\alpha\in\mathbb{R}^{N-L+1}$ such that \eqref{flat_DB}
		holds, where $\hat{\Psi}(\bar{u},\bar{y})$ is the stacked vector of the sequence $\{\hat{\Psi}_k(\bar{u},\bar{y})\}_{k=0}^{L-n-1}$.
		Since $\bar{u}$ is unknown, we substitute $\bar{u}=H_{L-n}(u)\alpha$. This can be done since one of the basis functions is the identity and hence $\hat{\Psi}(\bar{u},\bar{y})=H_{L-n}(\hat{\Psi}(u,y))\alpha$ implies $\bar{u}=H_{L-n}(u)\alpha$. By this substitution, we obtain the equivalent equation \eqref{flat_OM_eq}.
	\end{proof}\par
	\indent In order to apply Proposition 3 in case of noisy data, Equation \eqref{flat_OM_eq} is solved for $\alpha$ using a regularized least-squares optimization problem of the form
	\begin{equation}
		\hspace{-0.5mm}\min\limits_{\alpha}\;\norm*{\hspace{-0.5mm}\begin{bmatrix}H_{L-n}(\hat{\Psi}(u,y)) \\ H_L(y) \end{bmatrix}\hspace{-0.75mm}\alpha\hspace{-0.5mm}-\hspace{-1mm}\begin{bmatrix}\hat{\Psi}(H_{L-n}(u)\alpha,\bar{y})\\\bar{y}\end{bmatrix}\hspace{-0.5mm}}_2^2\hspace{-1.5mm} + \hspace{-0.5mm}\lambda\norm{\alpha}_2^2.
		\label{flat_opt_om}
	\end{equation}
	\begin{remark}
		If the basis functions ${\psi}_i$ are affine in the control input $u$, then problem \eqref{flat_opt_om} reduces to a quadratic program. Otherwise, it is a nonlinear optimzation problem.
	\end{remark}
	Similar to the discussion following Proposition 2, one can approximate the scalar product of the basis functions in \eqref{flat_opt_om} using a modified version of the kernels used in \eqref{kernels}. 
	\begin{align}
		&\Psi(u_i,y_{[i,i+n-1]})^\top \Psi(u_j,y_{[j,j+n-1]})=\notag\\&\quad\qquad u_iu_j+ e^{-\frac{\norm*{\small\begin{bmatrix}u_i\\y_{[i,i+n-1]}\end{bmatrix}-\begin{bmatrix}u_j\\y_{[j,j+n-1]}\end{bmatrix}}_2^2}{2\sigma^2}},\\
		&\Psi(h_i(u_{[i,i+N-L]})\alpha,\bar{y}_{[i,i+n-1]})^\top \Psi(u_j,y_{[j,j+n-1]})=\notag\\&\; (h_i(u_{[i,i+N-L]})\alpha) u_j+e^{-\frac{\norm*{\small\begin{bmatrix}h_i(u_{[i,i+N-L]})\alpha\\\bar{y}_{[i,i+n-1]}\end{bmatrix}-\begin{bmatrix}u_j\\y_{[j,j+n-1]}\end{bmatrix}}_2^2}{2\sigma^2}}.\notag
	\end{align}
	
	This modification allows for retrieving the input $\bar{u}$ by including ${\psi}_i(u,k)=u$ as one of the basis functions.\\
	\indent In this section, we outlined the procedure required to solve the data-driven simulation and output-matching control problems for SISO flat nonlinear systems. We emphasize again that in order to apply these results, only an upper bound on the system order is required as well as a set of basis functions ${\psi}_i$ such that \eqref{basis} is satisfied. If the latter is not available, one can instead use kernel methods to approximately solve these two problems. In the following section, we illustrate these results with numerical examples.
	\section{Numerical Examples}
	In this section, we consider two examples of flat nonlinear systems and apply the results of Propositions 2 and 3. Example 1 addresses the output matching control problem where the nonlinearity $v_k=\Phi(u_k,T_1^{-1}(\xi_k))$ in \eqref{basis} can be written as a finite expansion of known basis functions. Example 2 addresses the simulation problem for the more general case where the nonlinearity is only approximated using an infinite number of unknown basis functions.
	\begin{example}
		Consider the following flat system
		\begin{align*}
			x_1(k+1) &= x_2(k),\\
			x_2(k+1) &= u(k)\left(x_1^2(k)+2\right),\\
			y(k) &= x_1(k),
		\end{align*}
		which has a relative degree $d=2$. The system can be brought into \eqref{flat_NL} by choosing $\begin{bmatrix}\xi_1&\xi_2\end{bmatrix}^\top = T_1(x) = \begin{bmatrix}x_1&x_2\end{bmatrix}^\top$. We wish to solve the output matching control problem given a sinusoidal reference trajectory $\{\bar{y}_{\textup{\tiny{ref}},k}\}_{k=0}^{L-1}$ with $L=50$. We generate a persistently exciting input $\{u_k\}_{k=0}^{N-n-1}$ sampled from a uniform random distribution $U(-0.5,0.5)$ with $N=500$ and collect the corresponding output sequence $\{y_k\}_{k=0}^{N-1}$ from an open loop simulation. The output sequence is corrupted by additive noise sampled from $U(-0.025,0.025)$. Moreover, we use the following choice of basis functions $\Psi=\begin{bmatrix}u & u\xi_1 & u\xi_2 & \xi_1\xi_2 & u\xi_1^2 & u\xi_2^2\end{bmatrix}$ which contain the nonlinearity $v_k = u(k)\left(\xi_1^2(k)+2\right)$ in their span, i.e., \eqref{basis} is satisfied with $\text{a}^\top=\begin{bmatrix}2&0&0&0&1&0\end{bmatrix}$. Since the choice of basis functions is affine in $u$, the optimization problem in \eqref{flat_opt_om} is a quadratic program. Finally, we set the regularization parameter $\lambda=0.1$. Figure (1a) shows the estimated input by the proposed approach as well as the input computed by model-based output matching for validation purposes. In Figure (1b), we compare the true matched reference with the unknown system's response to the estimated input. It can be seen that the estimates are very close the model-based counterparts (considering the noise level), $\norm*{\bar{y}-\bar{y}_{\textup{\tiny{ref}}}}_2 = 0.2455,\,\norm*{\bar{u}-\bar{u}_{\textup{\tiny{ref}}}}_2 = 0.0708$. The estimation accuracy depends on the regularization parameter $\lambda$; small values do not mitigate the effect of the noise in $y$ while large values of $\lambda$ lead to poor performance, since $|\alpha|$ is small in that case.
	\end{example}
	\begin{figure}
		\begin{center}\includegraphics[scale=0.55]{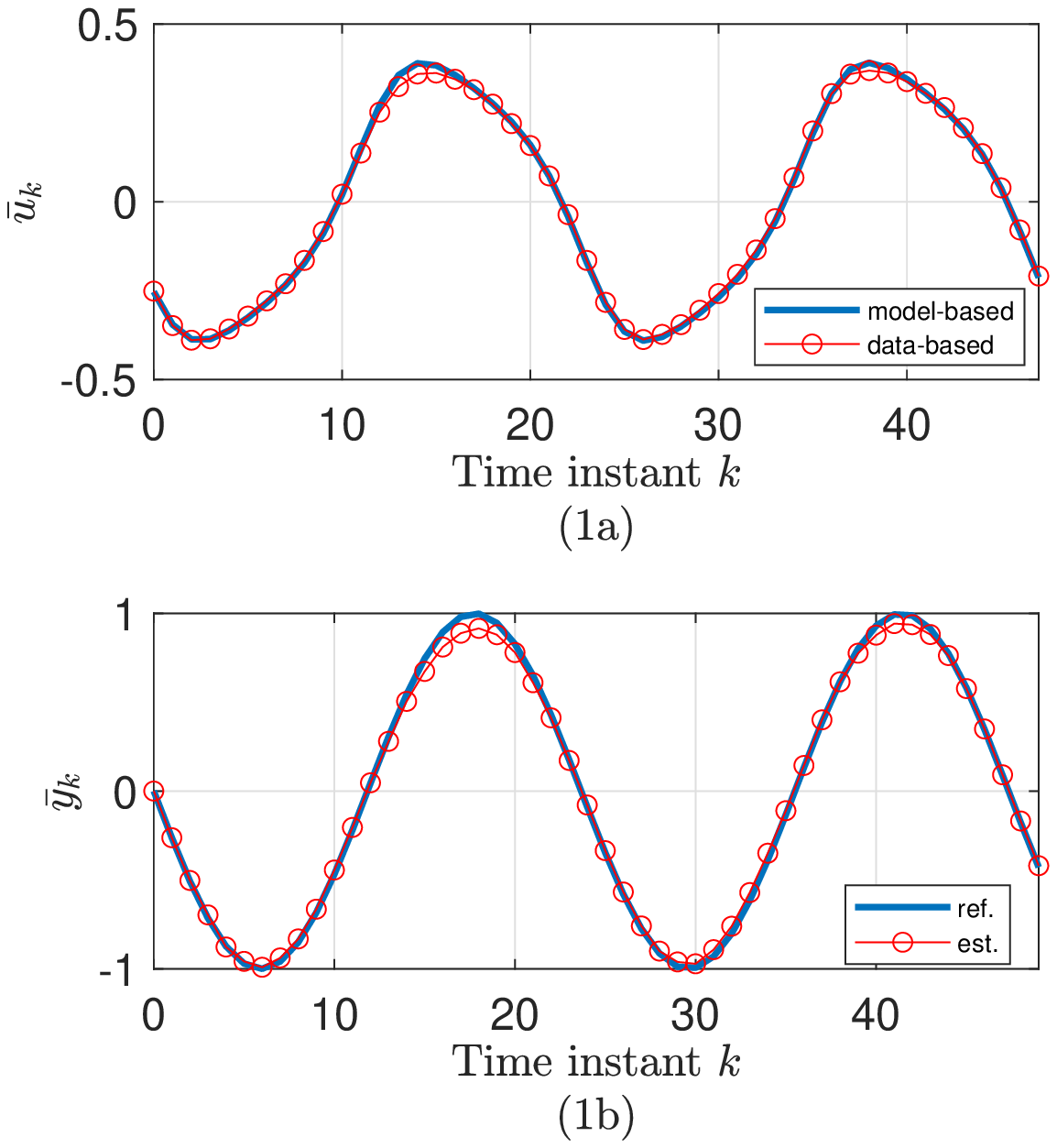}\end{center}
		\caption{(1a) Model-based and data-based estimated inputs and (1b) reference and estimated outputs, computed using the proposed data-driven output matching control approach.}
	\end{figure}
	\begin{example}
		Consider the following flat system
		\begin{align*}
			x_{1}(k+1) &= x_2(k),\\
			x_{2}(k+1) &= \frac{\sin(u)}{1+x_2^2(k)},\\
			y(k) &= x_1(k),
		\end{align*}
		which has a relative degree $d=2$. We wish to solve the simulation problem for a new input sequence $\{\bar{u}_k\}_{k=0}^{L-n-1}$ of length $L=50$, sampled from a uniform random distribution $U(-1,1)$. We generate an input sequence $\{u_k\}_{k=0}^{N-n-1}$ from $U(-1,1)$ with $N=750$ and collect the corresponding output sequence $\{y_k\}_{k=0}^{N-1}$ from an open loop simulation. The output sequence is corrupted by additive noise sampled from $U(-0.05,0.05)$. Finally, we use a bivariate Gaussian kernel as in \eqref{kernels} with $\sigma=1$ and solve the optimization problem in \eqref{sim_opt} in Matlab using CasADi \cite{casadi} with $\lambda=0.1$. Figure 2 shows the estimated output obtained from solving the data-based simulation problem along with the output of the system computed by model-based simulation. It can be seen from the results that the output is well approximated in this case by the proposed approach ($\norm*{\bar{y}-\bar{y}_{\textup{\tiny{true}}}}_2 = 0.3306$). In general, the estimation accuracy depends on several factors. The first is the length of the collected data $u,y$, and how the nonlinear system modes are excited by it. The second issue is whether a global minimum of the nonlinear optimization problem in \eqref{sim_opt} can (numerically) be found, since the dimension of $\alpha$ increases with $N$. Finally, a suitable tuning of the regularization parameter $\lambda$ and the kernel hyperparameter $\sigma$ is important. The problems posed by these factors are an open area of research that is currently being followed.
	\end{example}
	\begin{figure}
		\begin{center}\includegraphics[scale=0.55]{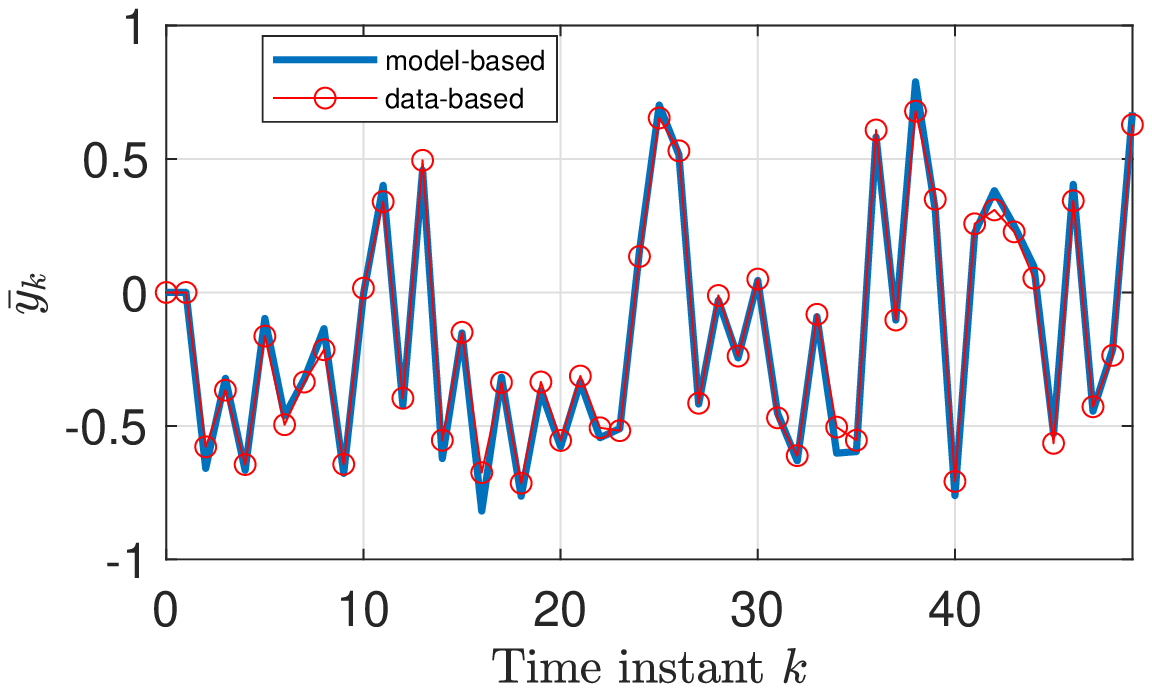}\end{center}
		\caption{Data-based estimated output computed using the proposed data-driven simulation approach with the Kernel method, and its model-based counterpart.}
	\end{figure}
	\section{Conclusions}
	In this paper, a data-based representation of SISO flat nonlinear systems was developed by leveraging the fact that they are exactly feedback linearizable. We have shown that a single, persistently exciting input-output trajectory spans the entire input-output trajectory space of the system, provided that the synthetic input is exactly expressed using a finite number of known basis functions. This representation was used to solve the simulation and output matching control problems in a purely data-based fashion. For a more practically relevant setting, we showed how to implicitly use infinitely many basis functions by using a kernel method. An interesting issue for future research is the quantification of the approximation error in data-based simulation and control when the nonlinearity is only approximated by the chosen basis functions.

\end{document}